\documentclass[unicode,runningheads,envcountsame,envcountsect,a4paper,dvipsnames,11pt]{llncs}
\usepackage[margin=20truemm]{geometry}

\usepackage[T1]{fontenc}

\usepackage{commands}
\usepackage{tikz}
\usetikzlibrary{fit, patterns.meta}

\begin{document}

\title{Finding a Maximum Common (Induced) Subgraph: \\ Structural Parameters Revisited\thanks{%
Partially supported
by JSPS KAKENHI Grant Numbers 
JP22H00513, 
JP24H00697, 
JP25K03076, 
JP25K03077, 
by JST, CRONOS, Japan Grant Number JPMJCS24K2, 
and
by JST SPRING, Grant Number JPMJSP2125. 
}}
\titlerunning{Finding an MC(I)S: Structural Parameters Revisited}
%
\author{
Tesshu Hanaka\inst{1}\orcidID{0000-0001-6943-856X} \and
Yuto Okada\inst{2}\orcidID{0000-0002-1156-0383} \and
Yota Otachi\inst{2}\orcidID{0000-0002-0087-853X} \and
Lena Volk\inst{3}\orcidID{0009-0004-3113-9205}
}

\authorrunning{T. Hanaka et al.}
%
\institute{%
Kyushu University, Fukuoka, Japan.
 \email{hanaka@inf.kyushu-u.ac.jp}
\and
Nagoya University, Nagoya, Japan.
 \email{pv.20h.3324@s.thers.ac.jp}, 
 \email{otachi@nagoya-u.jp}
\and
Technische Universit\"{a}t Darmstadt, Germany.
 \email{volk@mathematik.tu-darmstadt.de}
}
\maketitle              
\begin{abstract}

We study the parameterized complexity of the problems of finding a maximum common (induced) subgraph of two given graphs.
Since these problems generalize several NP-complete problems, 
they are intractable even when parameterized by strongly restricted structural parameters.
Our contribution in this paper is to sharply complement the hardness of the problems
by showing fixed-parameter tractable cases: 
both induced and non-induced problems parameterized by max-leaf number and by neighborhood diversity,
and the induced problem parameterized by twin cover number.
These results almost completely determine the complexity of the problems with respect to well-studied structural parameters.
Also, the result on the twin cover number presents a rather rare example where the induced and non-induced cases have different complexity.

\keywords{Maximum common (induced) subgraph \and Structural parameter \and Fixed-parameter tractability \and Twin cover \and Max-leaf number.}
\end{abstract}



\section{Introduction}

In this paper, we study structural parameterizations of {\probMCS} and {\probMCIS} and resolve a few new cases:
twin cover number, max-leaf number, and neighborhood diversity.
Our results almost completely determine the complexity of the problems in the hierarchy of well-studied graph parameters~\cite{SorgeW19},
while one important case of the induced version parameterized by cluster vertex deletion number remains unsettled.

Given graphs $G_{1}$, $G_{2}$ and an integer $h$,
{\probMCS} (\probMCSshort) asks whether there exists a graph $H$ with at least $h$ \emph{edges} such that both~$G_{1}$ and $G_{2}$ contain a subgraph isomorphic to~$H$.
{\probMCIS} (\probMCISshort) is a variant of {\probMCSshort} that asks for a common \emph{induced} subgraph with at least $h$ \emph{vertices}.
These problems are known to be intractable even in highly restricted settings.
Indeed, most of the known hardness results hold already for their special cases {\probSI} and {\probISI} as described later.

Given graphs $G$ and $H$, {\probSI} (\probSIshort) asks whether $G$ contains a subgraph isomorphic to $H$.
{\probISI} (\probISIshort) is a variant of {\probSIshort} that asks the existence of an \emph{induced} subgraph of $G$ isomorphic to $H$.
By setting $G_{1} = G$, $G_{2} = H$, and $h = |E(H)|$ ($h = |V(H)|$),
we can see that {\probSIshort} ({\probISIshort}) is a special case of {\probMCSshort} ({\probMCISshort}, respectively).
The problems {\probSIshort} and {\probISIshort} (and thus {\probMCSshort} and {\probMCISshort} as well) are NP-complete since they generalize many other NP-complete problems.
For example, if $H$ is a complete graph, then {\probSIshort} and {\probISIshort} coincide with \textsc{Clique}~\cite[GT19]{GareyJ79}.

\subsection{Background of the target setting}
\label{sec:background}
The problems {\probMCSshort}, {\probMCISshort}, {\probSIshort}, and {\probISIshort} have been studied extensively in many different settings. 
In this paper, we focus on the setting where \emph{both} input graphs have the same restriction on their structures.

We can see that {\probSIshort} and {\probISIshort} are NP-complete even if both $G$ and $H$ are path forests (i.e., disjoint unions of paths).
The hardness for {\probSIshort} appeared in an implicit way already in the book of Garey and Johnson~\cite[p.~105]{GareyJ79}.
The hardness for {\probISIshort} can be shown basically in the same way by adding a small number of vertices to each connected component of $G$ as separators~\cite{Damaschke90}.
Note that these hardness results imply NP-completeness of the case where both graphs are of 
bandwidth~$1$, feedback edge set number~$0$, and distance to path forest~$0$.
In this direction, we may ask the following question about their common special case, the max-leaf number.
\begin{quote}
  Q: \textit{Are {\probMCSshort} and {\probMCISshort} tractable when parameterized by max-leaf numbers of both input graphs?}
\end{quote}
By answering this question in the affirmative, we complement the hardness results in a sharp way.

For the case where both $G$ and $H$ are disjoint unions of complete graphs, the NP-completeness of {\probSIshort} follows directly from the path forest case
by replacing each connected component with a complete graph with the same number of vertices.
This implies that {\probSIshort} is NP-complete when both graphs are of twin cover number~$0$.
One of the main motivating questions in this work is whether such a reduction is possible for the induced case {\probISIshort}.
At least the same reduction does not work as we cannot embed two or more disjoint cliques into one clique as an induced subgraph.
Actually, it is not difficult to see that if both $G$ and $H$ are disjoint unions of complete graphs, {\probISIshort} is polynomial-time solvable.
On the other hand, {\probISIshort} is known to be NP-complete even on cographs~\cite{Damaschke90} (i.e., graphs of modular-width~$2$).
Now the question here is as follows.
\begin{quote}
  Q: \textit{Can we solve {\probMCISshort} efficiently when input graphs are close to a disjoint union of complete graphs in some sense?}
\end{quote}
We partially answer this question by presenting a fixed-parameter algorithm parameterized by twin cover number.

In the context of structural parameters, 
Abu-Khzam~\cite{Abu-Khzam14} showed that {\probMCISshort} is fixed-parameter tractable parameterized by $\vc(G_{1}) + \vc(G_{2})$,
where $\vc$ denotes the vertex cover number of a graph (see also \cite{Abu-KhzamBS17}).
This result was later generalized by Gima et al.~\cite{GimaHKKO22},
who showed that both {\probMCSshort} and {\probMCISshort} are fixed-parameter tractable parameterized by $\vi(G_{1}) + \vi(G_{2})$,
where $\vi$ denotes the vertex integrity of a graph.
On the other hand, Bodlaender et al.~\cite{BodlaenderHKKOOZ20} showed that {\probSIshort} is NP-complete on forests of treedepth~$3$.
Their proof can be easily modified to show that {\probISIshort} is NP-complete on the same class~\cite{GimaHKKO22}.
Note that $\td(G) \le \vi(G) \le \vc(G) + 1$ for every graph $G$, where $\td$ denotes the treedepth of a graph.
Bodlaender et al.~\cite{BodlaenderHKKOOZ20} also showed that 
{\probSIshort} is fixed-parameter tractable parameterized by $\nd(G) + \nd(H)$,
where $\nd$ denotes the neighborhood diversity of a graph, which is another generalization of vertex cover number
in the sense that $\nd(G) \le \vc(G) + 2^{\vc(G)}$ for every graph $G$~\cite{Lampis12}.
A natural question would be whether we can get the same results for {\probMCSshort} and {\probMCISshort}.
\begin{quote}
  Q: \textit{Are {\probMCSshort} and {\probMCISshort} tractable when parameterized by neighborhood diversity of both input graphs?}
\end{quote}
We show that the idea of the previous algorithm~\cite{BodlaenderHKKOOZ20} can be applied to {\probMCSshort} and {\probMCISshort} almost directly.


\subsection{Our results}

As mentioned above, our results can be summarized as follows.
\begin{enumerate}
  \item {\probMCSshort} and {\probMCISshort} are fixed-parameter tractable parameterized by $\ml(G_{1}) + \ml(G_{2})$, where $\ml$ denotes the max-leaf number of a graph.
  \item {\probMCISshort} is fixed-parameter tractable parameterized by $\tc(G_{1}) + \tc(G_{2})$, where $\tc$ denotes the twin cover number of a graph.
  \item {\probMCSshort} and {\probMCISshort} are fixed-parameter tractable parameterized by $\nd(G_{1}) + \nd(G_{2})$, where $\nd$ denotes the neighborhood diversity of a graph.
\end{enumerate}
Note that although the problems {\probMCSshort} and {\probMCISshort} are defined as decision problems, our positive results can be easily modified to output optimal solutions. 

See \cref{fig:graph-parameters} for the summary of the results.\footnote{
In \cref{fig:graph-parameters}, the abbreviations mean
clique-width ($\cw$),
treewidth ($\tw$),
pathwidth ($\pw$),
feedback vertex set number ($\fvs$),
feedback edge set number ($\fes$),
distance to path forest ($\dpf$),
bandwidth ($\bw$), 
modular-width ($\mw$),
shrub-depth ($\sd$), and
cluster vertex deletion number ($\cvd$).
}
Formal definitions of the graph parameters will be given as needed.

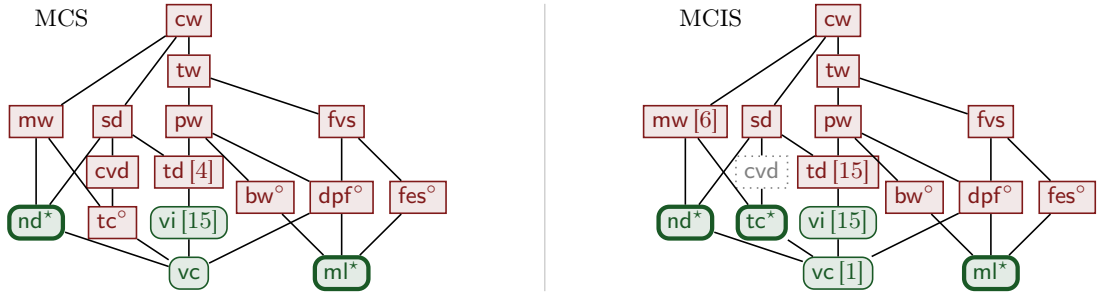
\begin{figure}[th]

{\centering

\definecolor{tkzdarkred}{rgb}{.5,.1,.1}
\definecolor{tkzdarkorange}{rgb}{.7,.25,0}
\definecolor{tkzdarkblue}{rgb}{.1,.1,.5}
\definecolor{tkzdarkgreen}{rgb}{.1,.35,.15}

\tikzset{npc/.style = {draw,semithick,rectangle,tkzdarkred,fill=tkzdarkred!10,align=center}}
\tikzset{whd/.style = {draw,semithick,rectangle,tkzdarkorange,align=center}}
\tikzset{fpt/.style = {draw,semithick,rectangle,rounded corners,tkzdarkgreen,fill=tkzdarkgreen!12,align=center}}
\tikzset{unk/.style = {draw,semithick,rectangle,dotted,gray,align=center}}

{}
\hfill
\begin{tikzpicture}[semithick, every node/.style={font=\small,text height=1.3ex, text depth = 0.2ex},scale=0.67]
  \node[npc] (cw) at (0cm, 5.0cm) {$\cw$};
  \node[npc] (tw) at (0cm, 4.0cm) {$\tw$};
  \node[npc] (pw) at (0cm, 3.0cm) {$\pw$};
  \node[npc] (td) at (0cm, 2.0cm) {$\td$\,\cite{BodlaenderHKKOOZ20}};
  \node[npc] (sd) at (-1.5cm, 3cm) {$\sd$};
  \node[fpt] (vi) at (0cm, 1.0cm) {$\vi$\,\cite{GimaHKKO22}};
  \node[fpt] (vc) at (0cm, 0.0cm) {$\vc$};
  \node[npc] (mw) at (-3cm, 3.0cm) {$\mw$};
  \node[fpt, ultra thick] (nd) at (-3cm, 1.0cm) {$\nd$\thispaper};
  \node[npc] (cvd) at (-1.5cm, 2cm) {$\cvd$};
  \node[npc] (tc) at (-1.5cm, 1cm) {$\tc$\easyobs};
  \node[npc] (fvs) at (3cm, 3cm) {$\fvs$};
  \node[npc] (fes) at (4.5cm, 1.5cm) {$\fes$\easyobs};
  \node[npc] (bw) at (1.5cm, 1.5cm) {$\bw$\easyobs};
  \node[npc] (dpf) at (3cm, 1.5cm) {$\dpf$\easyobs};
  \node[fpt, ultra thick] (mln) at (3cm, 0.0cm) {$\ml$\thispaper};

  \draw (cw) -- (mw) -- (nd) -- (vc);
  \draw (cw) -- (sd) -- (cvd) -- (tc) -- (vc);
  \draw (sd) -- (td);
  \draw (sd) -- (nd);
  \draw (mw) -- (tc);
  \draw (cw) -- (tw) -- (pw) -- (td) -- (vi) -- (vc);
  \draw (tw) -- (fvs);
  \draw (fvs) -- (fes) -- (mln);
  \draw (pw) -- (bw) -- (mln);
  \draw (pw) -- (dpf) -- (vc);
  \draw (fvs) -- (dpf) -- (mln);
  
  \node (caption) at (-2.5cm, 5cm) {{\probMCSshort}};
\end{tikzpicture} 
\hfill
\textcolor{lightgray}{\vrule}
\hfill
\begin{tikzpicture}[semithick, every node/.style={font=\small,text height=1.3ex, text depth = 0.2ex},scale=0.67]
  \node[npc] (cw) at (0cm, 5.0cm) {$\cw$};
  \node[npc] (tw) at (0cm, 4.0cm) {$\tw$};
  \node[npc] (pw) at (0cm, 3.0cm) {$\pw$};
  \node[npc] (td) at (0cm, 2.0cm) {$\td$\,\cite{GimaHKKO22}};
  \node[npc] (sd) at (-1.5cm, 3cm) {$\sd$};
  \node[fpt] (vi) at (0cm, 1.0cm) {$\vi$\,\cite{GimaHKKO22}};
  \node[fpt] (vc) at (0cm, 0.0cm) {$\vc$\,\cite{Abu-Khzam14}};
  \node[npc] (mw) at (-3cm, 3.0cm) {$\mw$\,\cite{Damaschke90}};
  \node[fpt, ultra thick] (nd) at (-3cm, 1.0cm) {$\nd$\thispaper};
  \node[unk] (cvd) at (-1.5cm, 2cm) {$\cvd$};
  \node[fpt, ultra thick] (tc) at (-1.5cm, 1cm) {$\tc$\thispaper};
  \node[npc] (fvs) at (3cm, 3cm) {$\fvs$};
  \node[npc] (fes) at (4.5cm, 1.5cm) {$\fes$\easyobs};
  \node[npc] (bw) at (1.5cm, 1.5cm) {$\bw$\easyobs};
  \node[npc] (dpf) at (3cm, 1.5cm) {$\dpf$\easyobs};
  \node[fpt, ultra thick] (mln) at (3cm, 0.0cm) {$\ml$\thispaper};

  \draw (cw) -- (mw) -- (nd) -- (vc);
  \draw (cw) -- (sd) -- (cvd) -- (tc) -- (vc);
  \draw (sd) -- (td);
  \draw (sd) -- (nd);
  \draw (mw) -- (tc);
  \draw (cw) -- (tw) -- (pw) -- (td) -- (vi) -- (vc);
  \draw (tw) -- (fvs);
  \draw (fvs) -- (fes) -- (mln);
  \draw (pw) -- (bw) -- (mln);
  \draw (pw) -- (dpf) -- (vc);
  \draw (fvs) -- (dpf) -- (mln);
  
  \node (caption) at (-2.5cm, 5cm) {{\probMCISshort}};
\end{tikzpicture}
\hfill
{}
}
\caption{The complexity of {\probMCSshort} (left) and {\probMCISshort} (right) when a structural parameter of both input graphs is bounded.
The normal rectangles and the rounded rectangles represent paraNP-complete cases and fixed-parameter tractable cases, respectively.
The results marked with {$\star$} are shown in this paper
and the ones with {$\circ$} are corollaries of the observations in \cref{sec:background}.
A connection between two parameters means that the one above is upper-bounded by a function of the one below (e.g., $\tw(G) \le \pw(G)$).}
\label{fig:graph-parameters}
\end{figure}

\subsection{Related results}
\label{sec:related-results}

Marx and Pilipczuk~\cite{MarxP14} studied the parameterized complexity of {\probSIshort}
and presented comprehensive results for many combinations of (possibly different) structural parameters of $G$ and $H$.
For the setting where both input graphs satisfy the same condition,
restricting the graph class that they belong to is another natural direction.
In this setting, Kijima et al.~\cite{KijimaOSU12} studied {\probSIshort}
and Heggernes et al.~\cite{HeggernesHMV15} studied {\probISIshort} 
both on interval graphs and related graph classes.

Jansen and Marx~\cite{JansenM15} considered {\probSIshort} in the setting where only $H$ belongs to a restricted hereditary graph class
and presented a dichotomy between randomized polynomial-time solvable cases and NP-complete cases.

If we restrict a structural parameter of $H$ only, almost all studied cases are known to be hard.
Both {\probSIshort} and {\probISIshort} are already W[1]-hard parameterized by $|V(H)|$ 
as \textsc{Clique} parameterized by the solution size~\cite{DowneyF99} reduces to this case.
We can see that {\probISIshort} is already NP-complete
when $\vc(H) = 0$ (\textsc{Independent Set}~\cite[GT20]{GareyJ79}) and
when $\ml(H) = 2$ (\textsc{Induced Path}~\cite[GT23]{GareyJ79}).
For {\probSIshort}, we can see that it  is NP-complete 
when $\nd(H) = 1$ and $\tc(H) = 0$ (\textsc{Clique}~\cite[GT19]{GareyJ79}),
when $\ml(H) = 2$ (\textsc{Hamiltonian Path}~\cite[GT39]{GareyJ79}), and 
when $\vi(H) = 3$ and $\tc(H) = 0$ (\textsc{Partition Into Triangles}~\cite[GT11]{GareyJ79}).
The only positive result known in this setting is that {\probSIshort} belongs to XP parameterized by $\vc(H)$;
indeed, a result of Bodlaender et al.~\cite[Theorem 14]{BodlaenderHJOOZ25} implies that {\probMCSshort} belongs to XP parameterized by $\min\{\vc(G_{1}), \vc(G_{2})\}$.
%


\section{Preliminaries}
\label{sec:pre}
We assume that the reader is familiar with the concept of fixed-parameter tractability.
See standard textbooks (e.g., \cite{CyganFKLMPPS15,DowneyF99,DowneyF13,FlumG06,Niedermeier06}) for the terms not defined in this paper.

Let $G$ be a graph. We denote by $\Delta(G)$ the maximum degree of $G$.
The set of connected component of $G$ is denoted by $\cc(G)$.
For a vertex $v \in V(G)$, the (\emph{open}) \emph{neighborhood} of $v$ is $N_{G}(v) = \{u \mid \{u,v\} \in E(G)\}$.
For a set $S \subseteq V(G)$, let $N_{G}(S) = \bigcup_{v \in S} N_{G}(v) \setminus S$.
We omit the subscript $G$ when the graph~$G$ is clear from the context.
For $S \subseteq V(G)$, we denote by $G[S]$ the subgraph of $G$ induced by $S$,
and we denote the subgraph $G[V(G) \setminus S]$ by $G - S$.

Let $G$ be a graph.
Two vertices $u, v \in V(G)$ are \emph{twins} if $N(u) \setminus \{v\} = N(v) \setminus \{u\}$.
That is, two vertices are twins if they have the same neighborhood when ignoring the adjacency between them.
Clearly, being twins is an equivalence relation.
A \emph{twin class} of $G$ is a maximal set of twins in $G$.
Observe that a twin class in a graph is either a clique or an independent set of the graph.

Let $G$ and $H$ be graphs.
An injective mapping $\eta \colon V(H) \to V(G)$ is a \emph{subgraph isomorphism} (resp.\ an \emph{induced subgraph isomorphism}) from $H$ to $G$
if for $u,v \in V(H)$, $\{u, v\} \in E(H)$ only if (resp.\ if and only if) $\{\eta(u), \eta(v)\} \in E(G)$.

For a positive integer $n$, we denote by $[n]$ the set of positive integers not greater than $n$, i.e., $[n] = \{1,\dots,n\}$.
Also, let $[n]_{0} = [n] \cup \{0\}$.


\section{{\probMCISshort} parameterized by twin cover number}

Let $G$ be a graph. An edge between twin vertices of $G$ is a \emph{twin edge} in $G$.
If an edge is not a twin edge, we call it a \emph{non-twin edge}.
A set $S \subseteq V(G)$ is a \emph{twin cover} of $G$ if every non-twin edge of $G$ has at least one of its endpoints in~$S$.
In other words, $S$ is a twin cover of $G$ if and only if $S$ is a vertex cover of~$G - F$, where $F$ is the set of twin edges of $G$.
Note that for a twin cover $S$, each connected component $K$ of $G - S$ is a complete graph and each vertex in~$K$ has the same neighborhood in $S$.

The \emph{twin cover number} of $G$, denoted $\tc(G)$, is the minimum size of a twin cover of $G$.
Since finding a minimum twin cover is fixed-parameter tractable parameterized by $\tc(G)$~\cite{Ganian15},
we assume that a minimum twin cover is given as part of the input when designing a fixed-parameter algorithm parameterized by~$\tc(G)$.

\begin{theorem}
\label{thm:mcis-tc}
{\probMCIS} is fixed-parameter tractable parameterized by $\tc(G_{1})+\tc(G_{2})$.
\end{theorem}
\begin{proof}
Let $\langle G_{1}, G_{2}, h \rangle$ be an instance of {\probMCISshort}
and $S_{1}$ and $S_{2}$ be minimum twin covers of $G_{1}$ and $G_{2}$, respectively. 
Let $p = \max\{\tc(G_{1}), \tc(G_{2})\}$.
We first guess the sets of vertices $S'_{1} \subseteq S_{1}$ and $S'_{2} \subseteq S_{2}$ that are included in a maximum common induced subgraph of~$G_{1}$ and $G_{2}$.
There are $2^{|S_{1} \cup S_{2}|}$ ($\le 2^{2p}$) candidates for this guess.
We update $G_{i}$ and $S_{i}$ as $G_{i} \coloneqq G_{i} - (S_{i} \setminus S'_{i})$ and $S_{i} = S'_{i}$ for each~$i \in \{1,2\}$.
Now our goal is to find a maximum common induced subgraph $H$ of~$G_{1}$ and $G_{2}$
with an induced subgraph isomorphism $\eta_{i}$ from $H$ to $G_{i}$ such that $S_{i} \subseteq \eta_{i}(V(H))$ for each $i \in \{1,2\}$.
Since $\tc(H) \le p$, we can guess from~$2^{O(p^{2})}$ candidates
the subgraph $H'$ of $H$ induced by a minimum twin-cover of $H$. Let~$T = V(H')$, i.e.,~$H' = H[T]$.

\paragraph{Guessing $\eta_{i}$ on $T$.}
Fix $i \in \{1,2\}$. We show that we can guess $\eta_{i}$ restricted to~$T$ from a small number of candidates.
We first guess the partition of $T$ into $T_{\to S_{i}} \coloneqq T \cap \eta_{i}^{-1}(S_{i})$ and $T_{\not\to S_{i}} \coloneqq T \setminus \eta_{i}^{-1}(S_{i})$ from~$2^{|T|}$ ($\le 2^{p}$) candidates.
Next we guess $\eta_{i}$ restricted to $T_{\to S_{i}}$ from $|S_{i}|^{|T_{\to S_{i}}|}$ ($\le p^{p}$) candidates.
To guess $\eta_{i}$ on $T_{\not\to S_{i}}$, the following claim is crucial for bounding the number of candidates.
\begin{myclaim}
\label{clm:tc-only-for-T}
If $u \in T$ with $\eta_{i}(u) \notin S_{i}$ and $K$ is the connected component of $G_{i} - S_{i}$ containing $\eta_{i}(u)$,
then $K \cap \eta_{i}(V(H)) \subseteq \eta_{i}(T)$.
\end{myclaim}
\begin{subproof}
[\cref{clm:tc-only-for-T}]
Suppose to the contrary that there is a vertex $v \notin T$ such that $\eta_{i}(v) \in K$.
Observe that $\eta_{i}(u)$ and $\eta_{i}(v)$ are twins in $G_{i}$ as they both belong to~$K$.
This implies that $u$ and $v$ are twins in~$H$ since $\eta_{i}$ is an induced subgraph isomorphism from $H$ to $G_{i}$.
By the minimality of $T$ as a twin cover, there is a non-twin edge $e$ of $H$ that $T \setminus \{u\}$ does not hit.
Observe that $e$ has $u$ as an endpoint, and so, let $e = \{u,w\}$.
Since $u$ and $v$ are twins, $w$ is a neighbor of $v$ as well.
On the other hand, since~$u$ and $w$ are not twins, $v$ and $w$ are not twins.
This contradicts the assumption that $T$ is a twin cover of~$H$
as $T$ does not hit the non-twin edge $\{v,w\}$.
\end{subproof}

Intuitively, \cref{clm:tc-only-for-T} means that once we decide to map a vertex in $T$ to a connected component~$K$ of $G_{i} - S_{i}$, 
then all  vertices mapped to $V(K)$ have to belong to $T$.
This allows us to guess $\eta_{i}$ on $T_{\not\to S_{i}}$ from at most $(p+1)^{p \cdot 2^{p} + p}$ candidates as follows.
\begin{enumerate}
  \item For $X \subseteq S_{i}$, let $\mathcal{K}_{X} = \{K \in \cc(G_{i} - S_{i}) \mid N_{G_{i}}(V(K)) = X\}$.

  \item For each $X \subseteq S_{i}$, we guess a vector in $[|T_{\not\to S_{i}}|]_{0}^{|T_{\not\to S_{i}}|}$ 
  that represents how~$\mathcal{K}_{X}$ contains vertices of~$\eta_{i}(T_{\not\to S_{i}})$.
  For example, $(2,1,0,0,\dots)$ means that 
  ``the class $\mathcal{K}_{X}$ contains three vertices of $\eta_{i}(T_{\not\to S_{i}})$ in total, two vertices in one clique and the other one in another clique.''
  There are at most $(|T_{\not\to S_{i}}| + 1)^{|T_{\not\to S_{i}}|}$ options for each $X \subseteq S_{i}$, and thus,
  at most $(|T_{\not\to S_{i}}|+1)^{|T_{\not\to S_{i}}| \cdot 2^{|S_{i}|}}$ ($\le (p+1)^{p \cdot 2^{p}}$) options in total.
  
  \item By \cref{clm:tc-only-for-T}, we can greedily use the smallest cliques in $\mathcal{K}_{X}$ that together satisfy the guessed vector.
  Furthermore, since each component of $G_{i}-S_{i}$ consists of twin vertices, we can arbitrarily pick the guessed number of vertices from the selected components.
  This gives us a complete guess of~$\eta_{i}(T_{\not\to S_{i}})$.
  
  \item Now we can guess $\eta_{i}$ on $T_{\not\to S_{i}}$ from $|T_{\not\to S_{i}}|^{|T_{\not\to S_{i}}|}$ ($\le p^{p}$) candidates as we already know the set $\eta_{i}(T_{\not\to S_{i}})$.
\end{enumerate}

So far, we guessed $\eta_{i}$ on $T_{\to S_{i}}$ and on $T_{\not\to S_{i}}$.
By combining them, we get $\eta_{i}$ on $T$ as $T_{\to S_{i}} \cup T_{\not\to S_{i}} = T$.
At this point, we reject the current guess if $\eta_{i}$ is not an isomorphism from \(H'\) (= \(H[T]\))
to $G_{i}[\eta_{i}(T)]$ for each $i \in \{1,2\}$.
Also, we reject the current guess if $\eta_{i}(T)$ is not a twin cover of $G_{i}[S_{i} \cup \eta_{i}(T)]$, which is an induced subgraph of $G_{i}[\eta_{i}(V(H))]$.
Note that $\eta_{i}(T)$ is not necessarily a twin cover of $G_{i}$ itself.

\paragraph{Cleaning up $G_{i}$.}
Now we show that when $\eta_{i}(T)$ is not a twin cover of $G_{i}$,
we can find a vertex in \(G_i\) not belonging to $\eta_{i}(V(H))$.
\begin{myclaim}
\label{clm:tc-clean-up}
If $\eta_{i}(T)$ is not a twin cover of $G_{i}$, then $G_{i}$ contains a non-twin edge between
a vertex in~$S_{i} \setminus \eta_{i}(T)$
and a vertex in $V(G_{i}) \setminus (S_{i} \cup \eta_{i}(T))$.
\end{myclaim}

\begin{subproof}
[\cref{clm:tc-clean-up}]
Let $\{u,v\}$ be a non-twin edge in $G_{i}$ with ${\{u,v\} \cap \eta_{i}(T) = \emptyset}$.
Since $S_{i}$ is a twin cover of $G_{i}$, at least one of~$u$ and $v$ belongs to $S_{i} \setminus \eta_{i}(T)$.
If exactly one of them belongs to $S_{i} \setminus \eta_{i}(T)$, then the other belongs to $V(G_{i}) \setminus (S_{i} \cup \eta_{i}(T))$, and thus the claim holds.
In the following, we assume that both~$u$ and $v$ belong to $S_{i} \setminus \eta_{i}(T)$.

Since $\{u,v\} \in E(G_{i})$ and $\eta_{i}(T)$ is a twin cover of $G_{i}[S_{i} \cup \eta_{i}(T)]$,
the vertices~$u$ and $v$ belong to the same connected component of $G_{i}[S_{i} \cup \eta_{i}(T)] - \eta_{i}(T)$.
Since $u$ and $v$ are not twins in $G_{i}$, there exists a vertex $w$ in $G_{i}$ adjacent to exactly one of them, say $u$.
Since $u, v \in S_{i} \setminus \eta_{i}(T)$ and $G[S_{i} \setminus \eta_{i}(T)]$ ($=G_{i}[S_{i} \cup \eta_{i}(T)] - \eta_{i}(T)$) is a disjoint union of complete graphs, $w \notin S_{i} \setminus \eta_{i}(T)$.
Since $\eta_{i}(T)$ is a twin cover of $G_{i}[S_{i} \cup \eta_{i}(T)]$, $u$ and $v$ have the same neighbors in $\eta_{i}(T)$, and thus $w \notin \eta_{i}(T)$.
Hence, we can conclude that $w \notin S_{i} \cup \eta_{i}(T)$.
This implies the claim since $\{u,w\}$ is a non-twin edge (because of $v$) with $u \in S_{i} \setminus \eta_{i}(T)$ and $w \notin S_{i} \cup \eta_{i}(T)$.
\end{subproof}

By \cref{clm:tc-clean-up}, if $\eta_{i}(T)$ is not a twin cover of $G_{i}$, then  we can find (in polynomial time) a non-twin edge $\{x,y\}$ in $G_{i}$ such that
$x \in S_{i} \setminus \eta_{i}(T)$ and ${y \in V(G_{i}) \setminus (S_{i} \cup \eta_{i}(T))}$.
Since $\eta_{i}(T)$ is a twin cover of $G_{i}[\eta_{i}(V(H))]$ and $S_{i} \subseteq \eta_{i}(V(H))$, it holds that $y \notin \eta_{i}(V(H))$.
Hence, we can safely remove such~$y$ from~$G_{i}$.
By exhaustively applying this reduction, we eventually obtain an induced subgraph $G_{i}'$ of~$G_{i}$ such that $\eta_{i}(V(H)) \subseteq V(G_{i}')$ 
and $\eta_{i}(T)$ is a twin cover for~$G_{i}'$.
We then update $G_{i}$ as~$G_{i} \coloneqq G_{i}'$ for each $i \in \{1,2\}$.

\paragraph{Finding a maximum solution under the guesses.}
The remaining task is to extend the already guessed parts of $H$ and $\eta_{i}$ to an entire solution.
Since we already know the correspondence between the twin covers $\eta_{1}(T)$ of $G_{1}$ and $\eta_{2}(T)$ of $G_{2}$,
we only need to \emph{match} the connected components of $G_{1} - \eta_{1}(T)$ and $G_{2} - \eta_{2}(T)$ in an optimal way.
Since each connected component of $G_{i} - \eta_{i}(T)$ is a complete graph, it can contain the image of at most one connected component of $H - T$.
This allows us to reduce the task to the maximum weighted bipartite matching problem as follows.

We construct a bipartite graph $B = (V_{1}, V_{2}; F)$, where $V_{1} = \cc(G_{1} - \eta_{1}(T))$ and $V_{2} = \cc(G_{2} - \eta_{2}(T))$.
Two vertices $K_{1} \in V_{1}$ and $K_{2} \in V_{2}$ are adjacent in $B$ if they have the \emph{same} adjacency to the corresponding twin covers; 
i.e., $\{K_{1}, K_{2}\} \in F$ if $\eta^{-1}_{1}(N_{G_{1}}(K_{1})) = \eta^{-1}_{2}(N_{G_{2}}(K_{2}))$. (Note that $N_{G_{i}}(K_{i}) \subseteq \eta_{i}(T)$.)
For $\{K_{1}, K_{2}\} \in F$, we set $w(\{K_{1}, K_{2}\}) = \min\{|V(K_{1})|, |V(K_{2})|\}$.
We call a vertex of $B$ corresponding to a connected component of $G_{i}[S_{i} \setminus \eta_{i}(T)]$ \emph{special}.
(Note that a connected component of $G_{i}[S_{i} \setminus \eta_{i}(T)]$ is a connected component of $G_{i} - \eta_{i}(T)$ as well.)
We compute a matching~$M$ of $B$ with the maximum weight under that condition that $M$ contains all special vertices.
Such a matching can be computed in polynomial time by using an algorithm for finding a maximum-weight degree-constrained subgraph~\cite{Gabow83}.
As we describe in the next paragraph, we can construct $H$ with~$|T| + w(M)$ vertices, where $w(M)$ is the total weight of the edges in $M$.
Thus, we set $|T| + w(M)$ to the optimal value under the current guess.

Given $M$ in the previous paragraph, we extend the already guessed parts of $H$ and $\eta_{i}$ as follows.
Let $\{K_{1}, K_{2}\} \in M$ with $K_{1} \in V_{1}$ and $K_{2} \in V_{2}$.
We add a complete graph $K_{H}$ of $\min\{|V(K_{1})|, |V(K_{2})|\}$ vertices into $H$
and make it adjacent to the subset $\eta^{-1}_{1}(N_{G_{1}}(K_{1}))$ ($= \eta^{-1}_{2}(N_{G_{2}}(K_{2}))$) of~$T$.
For $i \in \{1,2\}$, we extend $\eta_{i}$ by mapping $V(K_{H})$ to arbitrarily $\min\{|V(K_{1})|, |V(K_{2})|\}$ vertices in $K_{i}$.
After the extension, we have $|V(H)| = |T| + w(M)$.
Note that this solution satisfies the condition that $S_{i} \setminus \eta_{i}(T) \subseteq \eta_{i}(V(H))$ as $M$ contains all special vertices.
\end{proof}


\subsection{{\probMCISshort} parameterized by cluster vertex deletion number}

Given the fixed-parameter tractability parameterized by $\tc(G_{1})+\tc(G_{2})$, it is natural to consider the generalization to cluster vertex deletion number.
Although we do not have a complete answer, we present steps toward it by giving an XP algorithm and an FPT approximation scheme parameterized by $\cvd(G_{1})+\cvd(G_{2})$.

Let $G$ be a graph. A set $S \subseteq V(G)$ is a \emph{cluster vertex deletion set} of $G$ if every connected component of $G-S$ is a complete graph.
The \emph{cluster vertex deletion number} of $G$, denoted $\cvd(G)$, is the minimum size of a cluster vertex deletion set of $G$.
From their definitions, we have $\cvd(G) \le \tc(G)$ for every graph $G$.
Since finding a minimum cluster vertex deletion set is fixed-parameter tractable parameterized by~$\cvd(G)$~\cite{HuffnerKMN10},
we assume that a minimum cluster vertex deletion set is given when designing a fixed-parameter algorithm parameterized by $\cvd(G)$.

\begin{theorem}
\label{thm:mcis-cvd-xp}
{\probMCIS} belongs to $\mathrm{XP}$ parameterized by $\cvd(G_{1})+\cvd(G_{2})$.
\end{theorem}
\begin{proof}
Let $\langle G_{1}, G_{2}, h \rangle$ be an instance of {\probMCISshort}
and $S_{1}$ and $S_{2}$ be minimum cluster vertex deletion sets of $G_{1}$ and $G_{2}$, respectively.
Let $p = \max\{\cvd(G_{1}), \cvd(G_{2})\}$.
We first guess the vertices in $S_{1}$ and $S_{2}$ that are not included in a maximum common induced subgraph of $G_{1}$ and $G_{2}$.
We remove these vertices from $G_{1}$ and~$G_{2}$ and also from $S_{1}$ and $S_{2}$.
For simplicity, we still call the resulting graphs~$G_{1}$ and~$G_{2}$ and the resulting sets $S_{1}$ and $S_{2}$.
Next we guess, for each $i \in \{1,2\}$, the set $R_{3-i} \subseteq V(G_{3-i})$ that is matched with $S_{i}$.
That is, we are going to find a maximum common induced subgraph $H$ of $G_{1}$ and $G_{2}$ with an induced subgraph isomorphism $\eta_{i}$ from $H$ to $G_{i}$ 
such that $\eta_{i}^{-1}(S_{i}) = \eta_{3-i}^{-1}(R_{3-i})$ for each~$i \in \{1,2\}$.
Note that $\eta_{1}^{-1}(S_{1} \cup R_{1}) = \eta_{2}^{-1}(S_{2} \cup R_{2})$ and $|S_{i} \cup R_{i}| \le 2p$.
We guess the mapping~$\phi \coloneqq \eta_{2} \circ \eta_{1}^{-1}$ from $S_{1} \cup R_{1}$ to $S_{2} \cup R_{2}$. 
We reject the current guess if~$\phi$ is not an isomorphism from $G_{1}[S_{1} \cup R_{1}]$ to~$G_{2}[S_{2} \cup R_{2}]$.
The guesses made so far have at most $2^{2p} \cdot n^{2p} \cdot (2p)!$ candidates.

Since $S_{1} \cup R_{1}$ and $S_{2} \cup R_{2}$ are cluster vertex deletion sets of $G_{1}$ and $G_{2}$ and they are matched in the solution under consideration,
we can apply a similar matching phase as in the proof of \cref{thm:mcis-tc}.
We construct a bipartite graph with $V_{1} \coloneqq \cc(G_{1} - (S_{1} \cup R_{1}))$ as one color class and $V_{2} \coloneqq \cc(G_{2} - (S_{2} \cup R_{2}))$ as the other color class.
We add all possible edges between $V_{1}$ and $V_{2}$.
For $K_{1} \in V_{1}$ and  $K_{2} \in V_{2}$, we set the weight of the edge~$\{K_{1}, K_{2}\}$ to the maximum number of vertices that can be gained by matching $K_{1}$ and $K_{2}$.
More precisely, we set its weight to
\begin{align*}
  \sum_{X \subseteq S_{1} \cup R_{1}}
  \min\{
    &  |\{v_{1} \in K_{1} \mid N_{G_{1}}(v_{1}) \cap (S_{1} \cup R_{1}) = X\}|, \ 
       |\{v_{2} \in K_{2} \mid N_{G_{2}}(v_{2}) \cap (S_{2} \cup R_{2}) = \phi(X)\}|
    \}.
\end{align*}
We find a maximum-weight matching $M$ in this bipartite graph and set $|S_{1} \cup R_{1}| + w(M)$ to the optimal value under the current guess.
Correctness can be seen by observing that for $v_{1} \in K_{1}$ and $v_{2} \in K_{2}$, we can set $\eta_{1}^{-1}(v_{1}) = \eta_{2}^{-1}(v_{2})$ if and only if
$\phi(N_{G_{1}}(v_{1}) \cap (S_{1} \cup R_{1})) = N_{G_{2}}(v_{2}) \cap (S_{2} \cup R_{2})$.
\end{proof}

As an additional remark, we show that {\probMCISshort} parameterized by $\cvd(G_{1})+\cvd(G_{2})$ admits an FPT approximation scheme.
For graphs $G_{1}$ and $G_{2}$, let $\mcis(G_{1},G_{2})$ be the number of vertices in a maximum common induced subgraph of $G_{1}$ and $G_{2}$.
For $r \le 1$, an \emph{$r$-approximation solution} for {\probMCISshort} on $G_{1}$ and~$G_{2}$
is a common induced subgraph of $G_{1}$ and $G_{2}$ with at least $r \cdot \mcis(G_{1},G_{2})$ vertices.
The \emph{vertex integrity} of $G$, denoted $\vi(G)$, is defined as $\vi(G) = \min_{S \subseteq V(G)} (|S| + \max_{C \in \cc(G - S)} |C|)$.
\begin{theorem}
\label{thm:mcis-cvd-fpt-as}
Given graphs $G_{1}$ and $G_{2}$ and a number $\epsilon \in (0,1)$,
finding a $(1-\epsilon)$-approximation solution for {\probMCISshort} is fixed-parameter tractable parameterized by $(\cvd(G_{1})+\cvd(G_{2}))/\epsilon$.
\end{theorem}
\begin{proof}
Let $\langle G_{1}, G_{2}, k \rangle$ be an instance of {\probMCISshort}
and $S_{1}$ and $S_{2}$ be minimum cluster vertex deletion sets of $G_{1}$ and $G_{2}$, respectively.
Let $p = \max\{\cvd(G_{1}), \cvd(G_{2})\}$.
Let $G_{i}^{*}$ be an induced subgraph of $G_{i}$ obtained by removing vertices until every twin class has size at most $2p/\epsilon$.
We can see that~$\vi(G_{i}^{*}) \le p + 2^{p} \cdot (2p/\epsilon)$ as each connected component of $G_{i}-S_{i}$ can be partitioned into at most $2^{|S_{i}|}$ ($\le 2^{p}$) twin classes.
Thus, $\mcis(G_{1}^{*}, G_{2}^{*})$ can be computed using the fixed-parameter algorithm for {\probMCISshort} parameterized by vertex integrity~\cite{GimaHKKO22}.
We output the solution corresponding to the larger one of $\mcis(G_{1}^{*}, G_{2}^{*})$ and $\mcis(G_{1}-S_{1}, G_{2}-S_{2})$.
Note that  $\mcis(G_{1}-S_{1}, G_{2}-S_{2})$ can be computed in polynomial time by the algorithm in \cref{thm:mcis-tc} or \cref{thm:mcis-cvd-xp} as $\tc(G_{i} - S_{i}) = \cvd(G_{i} - S_{i}) = 0$.

We show that the output is a $(1-\epsilon)$-approximation solution.
Assume that $\mcis(G_{1}-S_{1}, G_{2}-S_{2}) < (1-\epsilon) \cdot \mcis(G_{1}, G_{2})$ (otherwise, we are done).
Since $\mcis(G_{1}, G_{2}) \le \mcis(G_{1}-S_{1}, G_{2}-S_{2}) + |S_{1}|+|S_{2}| = \mcis(G_{1}-S_{1}, G_{2}-S_{2}) + 2p$,
we have $\mcis(G_{1}, G_{2}) < (1-\epsilon) \cdot \mcis(G_{1}, G_{2})+2p$,
and thus $\mcis(G_{1}, G_{2}) < 2p/\epsilon$.
This implies that $\mcis(G_{1}, G_{2}) = \mcis(G_{1}^{*}, G_{2}^{*})$
since a maximum common induced subgraph of $G_{1}$ and $G_{2}$ uses at most $\mcis(G_{1}, G_{2})$ ($< 2p/\epsilon$) vertices in each twin class of $G_{i}$.
\end{proof}


\section{{\probMCSshort} and {\probMCISshort} parameterized by max-leaf number}

The \emph{max-leaf number} of a connected graph $G$, denoted $\ml(G)$, is the maximum number of leaves in a spanning tree of $G$.
For a disconnected graph $G$, we define its max-leaf number as the sum of the max-leaf number of its connected components;
that is, $\ml(G) = \sum_{C \in \cc(G)} \ml(C)$.
From the definition of $\ml(G)$, we can see that $\ml(G) \ge |\cc(G)|$
and $\ml(G) \ge \Delta(G)$. (To see the latter, consider a BFS tree rooted at a vertex of the maximum degree.)

In this section, we show that both {\probMCSshort} and {\probMCISshort} are fixed-parameter tractable parameterized by the max-leaf number of both graphs.
\begin{theorem}
	\label{thm:mcs-mln}
	{\probMCS} is fixed-parameter tractable parameterized by $\ml(G_{1}) + \ml(G_{2})$.
\end{theorem}
\begin{theorem}
	\label{thm:mcis-mln}
	{\probMCIS} is fixed-parameter tractable parameterized by $\ml(G_{1}) + \ml(G_{2})$.
\end{theorem}

In our algorithms, we do not need to compute a spanning tree with the maximum number of leaves
(although it is actually fixed-parameter tractable parameterized by the number of leaves~\cite{Zehavi18}).
Instead, we use a polynomial-time computable structure described below.

By~\(V_{\neq 2}(G)\), we denote the non-degree-2 vertices of \(G\).
It is known that, using the result of Kleitman and West~\cite{KleitmanW91},
the number of non-degree-2 vertices can be bounded from above by a linear function of max-leaf number (see e.g.,~\cite{FellowsLMMRS09}).
To be self-contained, we include a proof of the following statement in the appendix.
(Indeed, we prove a slightly stronger bound of $4\ml(G) - 6$, which is tight. See \cref{lem:deg-2-bound-mln}.)

\begin{lemma}[Folklore]
\label{lem:deg-2-bound-mln-weak}
For every graph $G$, $|V_{\ne 2}(G)| \le 4 \ml(G)$.
\end{lemma}

We call a trail (i.e., a walk in which no edges are repeated) in a graph \(G\) a \emph{maximal degree-2 trail} if all internal vertices are of degree-2 in~\(G\) and both endpoints are non-degree-2 vertices in \(G\).
Further, if \(G\) has simple cycles as connected components (which we call \emph{isolated cycles} in the following), we also consider them as maximal degree-2 trails by selecting an arbitrary vertex in the cycle as its endpoint.
Let~\(\mathcal{T}_2(G)\) be the set of maximal degree-2 trails.
Note that an element of~\(\mathcal{T}_2(G)\) is either a path or a cycle in \(G\). 
Note also that we can compute $\mathcal{T}_2(G)$ in polynomial time.

Using \cref{lem:deg-2-bound-mln-weak}, we can bound the number of maximal degree-2 trails as follows.
\begin{lemma}
\label{lem:num-maximal-deg2-trails}
For every graph $G$, $|\mathcal{T}_{2}(G)| \le 2 \ml(G)^{2}$.
\end{lemma}
\begin{proof}
Let $\ell = \ml(G)$.
First assume that $G$ is a connected graph.
If $G$ is a cycle, then the lemma holds as $|\mathcal{T}_{2}(G)| = 1$ and $\ell = 2$.
In the following, assume that $G$ is not a cycle. Then, each endpoint of a trail in~$\mathcal{T}_{2}(G)$ is a non-degree-2 vertex.
Thus, a trail in~$\mathcal{T}_{2}(G)$ contributes exactly~$2$ to the degree sum of the non-degree-2 vertices.
This gives us the desired bound as follows:
\[
 2 |\mathcal{T}_{2}(G)| = \sum_{v \in V_{\ne 2}(G)} \deg_{G}(v) \le |V_{\ne 2}(G)| \cdot \Delta(G) \le 4\ell^{2},
\]
where the last inequality holds by \cref{lem:deg-2-bound-mln-weak} and the fact $\Delta(G) \le \ell$.

Next assume that $G$ is disconnected. In this case, we can apply the lemma to each connected component as follows:
\[
  |\mathcal{T}_{2}(G)| = \sum_{C \in \cc(G)} |\mathcal{T}_{2}(C)| \le \sum_{C \in \cc(G)} 2 (\ml(C))^{2} \le 2 \Bigg(\sum_{C \in \cc(G)} \ml(C) \Bigg)^{2} = 2 \ell^{2}.
\]
This completes the proof.
\end{proof}

Although a subgraph (or an induced subgraph) of a graph of bounded max-leaf number may have unbounded max-leaf number in general,
the next lemma shows that a maximum common (induced) subgraph of graphs with bounded max-leaf number always has bounded max-leaf number.

\begin{lemma}
\label{lem:max-subg-bounded-mln}
For graphs $G_{1}$ and $G_{2}$ with max-leaf number at most $\ell$,
every maximum common (induced) subgraph~$H$ has max-leaf number at most $24 \ell^5$.
\end{lemma}

\begin{proof}
Since every subtree of a graph can be extended to a spanning tree of the graph with at least as many leaves, every connected subgraph of a graph with max-leaf number \(\ell\), has max-leaf number not greater than~\(\ell\). Thus, it suffices to show that for graphs \(G_1\) and \(G_2\) with max-leaf number \(\ell\), there is a maximum common (induced) subgraph \(H\) with at most \(24 \ell^4\) connected components. 

We begin with the non-induced case. Assume \(H\) has more than \(24 \ell^4\) connected components. 
Then, by \cref{lem:num-maximal-deg2-trails}, at least \(12 \ell^2\) ($\le 24 \ell^{4} / |\mathcal{T}_{2}(G_{i})|$) many connected components of \(H\) intersect the same degree-2 trail \(T_1\) in \(G_1\).
 Further, from these \(12 \ell^2\) connected components, at least six also intersect the same degree-2 trail~\(T_2\) in~\(G_2\). Thus, in the embedding at least four of the components are completely contained in~\(T_2\). In particular, they are paths. From these four, again at least two are completely contained in \(T_1\) when embedded to \(G_1\). We call these two connected components of \(H\)~\(p\) and \(p'\). Since~\(p\) and \(p'\) are completely contained in \(T_1\), we find an embedding from~\(H\) to \(G_1\) such that \(p\) and \(p'\) are neighboring in~\(T_1\) (by rearranging the embedding of the components of~\(H\) that are completely embedded in~\(T_1\)). Analogously, we find an embedding of \(H\) to~\(G_2\), where \(p\) and~\(p'\) are neighboring in \(T_2\). Thus,~\(H\) together with the additional edge connecting the paths \(p\) and~\(p'\) is also a common subgraph of \(G_1\) and~\(G_2\). This contradicts the maximality of~\(H\).

For the induced case we proceed analogously: By the same argument, we find two connected components~\(p\) and~\(p'\) of \(H\), with the property, that they are paths and there is an embedding of \(H\) to \(G_1\) in which \(p\) and \(p'\) are in the same maximal degree-2 trail of \(G_1\) and of distance-2 and the same holds for the embedding of~\(H\) to \(G_2\). First, note that the vertex between \(p\) and \(p'\) in \(G_1\) (resp.~\(G_2\)) is not already contained in the embedding of \(H\) to \(G_1\) (resp.~\(G_2\)), as~\(H\) is an induced subgraph. But now, again, we find a larger common induced subgraph of \(G_1\) and \(G_2\), namely by taking \(H\) and combining~\(p\) and \(p'\) to a single path with an additional vertex.
\end{proof}

The \emph{smoothing} of a graph \(G\) is the graph obtained by repeatedly deleting a vertex of degree 2 and its incident edges and then adding an edge between its neighbors, which may be the same vertex, until no degree-2 vertices remain except for isolated loops. Note that the smoothing of a simple graph may have loops and multi-edges. We begin by proving \cref{thm:mcs-mln} and then adapt the proof to show \cref{thm:mcis-mln}.

\begin{proof}[\cref{thm:mcs-mln}]
	Let \(\langle G_{1}, G_{2}, h \rangle\) be an instance of {\probMCS} and let \(\ell= \max \{\ml(G_1),\ml(G_2)\}\).

First of all, we guess the smoothing \(S\) of a maximum common subgraph of~\(G_1\) and~\(G_2\).
	By \cref{lem:max-subg-bounded-mln}, every maximum common subgraph of \(G_1\) and \(G_2\) has max-leaf number bounded by \(24 \ell^5\). Then, \cref{lem:deg-2-bound-mln-weak} yields that a maximum common subgraph of \(G_1\) and \(G_2\) has at most~\(96\ell^5\) non-degree-2 vertices. Further, it has maximum degree at most \(24 \ell^5\) and contains at most \(24 \ell^5\) many isolated loops.
Thus, as \(S\), we guess a graph on at most \(96\ell^5\) vertices with maximum degree \(24 \ell^5\), possibly with multi-edges and loops, but without any degree-2 vertices, and additionally with up to~\(24 \ell^5\) isolated loops. Note that we can assume~\(S\) to have no isolated vertices, as it is the smoothing of a maximum common subgraph.
Clearly, the number of possible options for $S$ is bounded by a function depending only on $\ell$.

	Before we proceed to the next step, we need to introduce some notation. We call an alternating sequence of elements in \(V_{\neq 2}(G_i)\) and \(\mathcal{T}_2(G_i)\) \emph{valid}, if 
	\begin{itemize}
		\item it contains at least one element of \(\mathcal{T}_2(G_i)\),
		\item it does not visit an element of $V_{\ne 2}(G_{i})$ multiple times with the only exception that it may start and end at the same vertex,
		\item every element \(T\) of \(\mathcal{T}_2(G_i)\) in the sequence is in between its two endpoints, or at the beginning or end of the sequence and next to one of its endpoints, or the only element in the sequence.
	\end{itemize} 
	Note that the second and third condition together imply that all trails, except for the first and last ones, in such a valid sequence are pairwise disjoint as well. In particular, a valid sequence corresponds to a path or a cycle in \(G_i\). In such a sequence, we call all elements besides the first and last ones \emph{inner}. Inner trails in a sequence always correspond to paths in \(G_i\), with the only exception of a sequence with three elements starting and ending at the same vertex and a degree-2 cycle between them.
	Further, in the following, we do not distinguish between a sequence and the same sequence in reversed order. 
	
	Now, we guess a mapping \(\eta_i\) from \(S\) to \(G_i\) with the following properties:
	\begin{itemize}
		\item For each vertex \(v \in V(S)\) with \(\deg(v)>2\), guess a vertex \(\eta_i(v) \in V(G_i)\) with \(\deg(\eta_i(v)) \geq \deg(v)\).
		\item For each vertex \(v \in V(S)\) with \(\deg(v)=1\), guess either a vertex \(\eta_i(v) \in V(G_i)\) with \(\deg(\eta_i(v)) \neq 2\), or guess a maximal degree-2 trail \(\eta_i(v) \in \mathcal{T}_2(G_i)\).
		\item For each isolated loop \(l=\{u,u\} \in E(S)\), guess either a vertex \(\eta_i(u) \in V(G_i)\) with \(\deg(\eta_i(u)) > 2\), or guess an isolated cycle \(\eta_i(u)\) in \(G_i\).
		\item Ensure that until now each vertex in \(V(G_i)\) is guessed at most once, i.e., \(|\eta_i^{-1}(u)| \leq 1\) for all \(u \in V(G_i)\).
		\item For each edge \(e \in E(S)\) with endpoints \(u\) and \(v\) (note that there could be multiple edges with this property, and it might hold \(u=v\)), guess a valid alternating sequence \(\eta_i(e)\) of elements in~\(V_{\neq 2}(G_i)\) and \(\mathcal{T}_2(G_i)\) starting and ending with \(\eta_i(u)\) and \(\eta_i(v)\).
		\item Ensure that for all edges \(e,e' \in E(S)\) the inner elements of \(\eta_i(e)\) and \(\eta_i(e')\) are pairwise disjoint.
		\item Ensure that for every edge \(e \in E(S)\) all inner elements of \(\eta_i(e)\) are pairwise disjoint to all elements in \(\eta_i(V(S))\).
	\end{itemize}
	Since the numbers of vertices and edges in \(S\), and the numbers of non-degree-2 vertices, maximal degree-2 trails, and connected components in \(G_i\) is bounded by a function depending only on $\ell$, the number of guesses is bounded by a function of $\ell$ as well. 

	Further, for the smoothing of every maximum common subgraph of \(G_1\) and~\(G_2\), there exists a mapping with the properties above, as it is induced by the corresponding embeddings. Vice versa, this mapping later gives rise to an embedding of the resulting graph to \(G_1\) and \(G_2.\)
	
	Before we proceed, note that if for an edge \(e\) with endpoints \(u\) and \(v\), we have~\(\eta_i(u)=\eta_i(v) \in V(G_i)\), then \(\eta_i(e)\) has an inner element. Further, if \(\eta_i(e)\) has no inner elements, then at least one of \(\eta_i(u)\) and~\(\eta_i(v)\) belongs to \(\mathcal{T}_2(G_i)\). If both of them belong to \(\mathcal{T}_2(G_i)\), then they coincide and \(\eta_i(e)\) only contains this one element.
	
	Based on the graph \(S\) and the mappings \(\eta_i\), we can now reduce the problem of finding a maximum common subgraph of \(G_1\) and \(G_2\) to \textsc{Integer Linear Programming} (ILP) with bounded number of variables. It is known that ILP parameterized by the number of variables is fixed-parameter tractable~\cite{Lenstra83}
even with a linear objective function to maximize (see e.g., \cite{FellowsLMRS08}).
	
	For each edge \(e \in E(S)\) introduce a variable \(l_e\). This variable describes the length of the path in a maximal common subgraph corresponding to the edge~\(e\) in its smoothed graph~\(S\).
	Further, for every~\(v \in V(S)\) with \(\deg(v)=1\) such that \(\eta_i(v) \in \mathcal{T}_2(G_i)\) is a path with distinct endpoints \(x_i\) and~\(y_i\), introduce two variables~\(v_{x_i}\) and~\(v_{y_i}\). These variables describe where on the path \(\eta_i(v)\) the vertex~\(v\) is mapped in an embedding, giving the distances from $v$  to the endpoints \(x_i\) and~\(y_i\). Analogously, for every \(v \in V(S)\) with \(\deg(v)=1\) such that \(\eta_i(v) \in \mathcal{T}_2(G_i)\) is a cycle with endpoint \(z_i\), introduce a variable~\(d_{v,z_i}\). This variable describes where on the cycle~\(\eta_i(v)\) the vertex \(v\) is mapped in an embedding, by giving the distance from $v$ to \(z_i\).

	Now the objective is to maximize \(\sum_{e \in E(S)} l_e\) under the following constraints.
	Note that some of the following constraints are not linear (like taking one of two values or having an equality to an absolute value). We can handle this issue by guessing the correct (linear) constraints corresponding to an optimal solution since the numbers of variables and constraints depend only on~$\ell$.
	\begin{enumerate}
		\item\label{cons:1} For all \(l_e, v_{x_i}, v_{y_i}, d_{v,z_i}\):
		\[l_e, v_{x_i},v_{y_i},d_{v,z_i} \geq 1.\]
		
		\item\label{cons:2} For all \(v_{x_i},v_{y_i}\):
		\[v_{x_i}+v_{y_i} = |E(\eta_i(v))|.\]
		
		\item\label{cons:3} For all \(d_{v,z_i}\):
		\[d_{v,z_i} \leq \lfloor |E(\eta_i(v))|/ 2 \rfloor.\]
		
		\item For all \(e \in E(S)\) with endpoints \(u \neq v\) such that \(\eta_i(e)\) contains no inner elements
		(note that at most one of $\eta_{i}(u)$ and $\eta_{i}(v)$ belongs to $V(G_{i})$):
		\begin{itemize}
			\item If \(\eta_i(v) \in V(G_i)\) (the case of \(\eta_i(u) \in V(G_i)\) is symmetric), then \(\eta_i(u) \in \mathcal{T}_2(G_i)\) is a path or a cycle.
			\begin{itemize}
				\item If \(\eta_i(u)\) is a path with endpoints \(\eta_i(v)\) and \(y_i\):
				\[l_e=u_{\eta_i(v)}.\]
				\item If \(\eta_i(u)\) is a cycle with endpoint \(\eta_i(v)\):
				\[l_e \in \{d_{u,\eta_i(v)}, |E(\eta_i(u))|-d_{u,\eta_i(v)}\}.\]
			\end{itemize}
			\item If \(\eta_i(u),\eta_i(v) \notin V(G_i)\), then \(\eta_i(v)=\eta_i(u)\in \mathcal{T}_2(G_i)\) is a path or a cycle.
			\begin{itemize}
				\item If \(\eta_i(u)\) is a path with endpoints \(x_i\) and \(y_i\):
				\[l_e = |u_{x_i} - v_{x_i}|.\]
				\item If \(\eta_i(u)\) is a cycle with endpoint \(z_i\): 
				\[l_e = |d_{u,z_i} - d_{v,z_i}|.\]
			\end{itemize}
		\end{itemize}
		
		\item\label{const:loops} For all loops \(e \in E(S)\) with endpoint \(u\) such that \(\eta_i(e)\) contains no inner elements (note that this can only happen if \(e\) is an isolated loop for which~\(\eta_i(u)\) is an isolated cycle \(c\) in \(G_i\)):
			\[l_e = |E(c)|.\]
		\item For all \(e \in E(S)\) with endpoints \(u,v\) such that \(\eta_i(e)\) contains inner elements:
		\[l_e = a + b + \sum \limits_{p \text{ inner path in } \eta_i(e)} |E(p)|, \quad \text{ where}\]
		 \[
		a =
		\begin{cases}
			0 & \text{if } \eta_i(e) \text{ starts with a vertex in } G_i, \\
			v_{y_i}  & \text{if } \eta_i(e) \text{ starts with } (\eta_i(v),y_i, \dots) \text{ and } \eta_i(v) \text{ is a path}, \\
			 d & \text{if } \eta_i(e) \text{ starts with }(\eta_i(v),z_i, \dots) \text{ and } \eta_i(v) \text{ is a cycle,}\\
			 & \quad \text{and } d \in \{d_{v,z_i},|E(\eta_i(v))|-d_{v,z_i}\},
		\end{cases}
		\]
		and \(b\) is defined analogously for the end of the sequence.
		\item\label{cons:length} For every maximal degree-2 trail \(T_i \in \mathcal{T}_2(G_i)\), which is not an inner element of some \(\eta_i(e)\): Let \(x_i\) and \(y_i\) be the endpoints of \(T_i\). Then, there is at most one edge \(e \in E(S)\) such that \(x_i \in \eta_i(e)\) and analogously there is at most one edge \(e' \in E(S)\) such that \(y_i \in \eta_i(e')\).
		\[|E(T_i)| \geq (a+1) + b + \sum_{f \in E(S) \text{ with } \eta_i(f)=(T_i)} (l_{f}+1), \quad \text{ where}\]
		\[
		a =
		\begin{cases}
			0 & \text{if no such edge } e \text{ exists,} \\
			v_{x_i}  & \text{if } T_i \text{ is a path and } v \text{ is the endpoint of } e \text{ with } \eta_i(v)=T_i \\
			& \text{and } \eta_i(e) \text{ starts or ends with } (T_i, x_i, \dots) \text{ or } (\dots, x_i,T_i), \\
			d & \text{if } T_i \text{ is a cycle and } v \text{ is in the case as above and }\\
			& d \in \{d_{v,x_i},|E(\eta_i(v))|-d_{v,x_i}\},
		\end{cases}
		\]
		and \(b\) is defined analogously for \(e'\). Note that in case \(T_i\) is a cycle, \(b\) is~$0$.
	\end{enumerate}

	In the first condition, we require \(l_e \geq 1\), as the subdivision of any edge is at least of length 1 (namely in case the edge was not subdivided at all). Further, we demand~\(v_{x_i},v_{y_i},d_{v,z_i} \geq 1\), as in case of lengths~0 we can instead consider the initial guess of \(\eta_i\) where we guessed the corresponding vertex for~\(\eta_i(v)\) instead of the maximal degree-2 trail. The second and third conditions ensure, that the position of the vertex \(v\) in the embedding is well-defined. The fourth and fifth conditions ensure that the length \(l_e\) of the edges which get completely embedded into one maximal {degree-2} trail in \(G_i\) is well-defined. Analogously, the sixth condition ensures that the length \(l_e\) of the remaining edges (i.e., the ones which get embedded into more than one maximal degree-2 path) is well-defined. The seventh condition then ensures, that every maximal {degree-2} trail in~\(G_i\) is long enough for all parts which get mapped there. 
	
	To see that the seventh condition suffices to ensure that all parts which are mapped to a maximal {degree\hbox{-}2} trail~\(T_i\) can get properly embedded, note that~\(T_i\) is in the image of at most two edges~\({e, e' \in E(S)}\) such that \(\eta_i(e), \eta_i(e')\) do not only contain \(T\). Then, the requirement on the pairwise disjointness of the valid sequences already ensures, that for both of these edges different endpoints of the trail are used (or, in case the trail is a cycle only one such edge exists). Further, all edges which get completely mapped to \(T_i\) correspond to paths in \(G_i\) and hence, it only remains to ensure that the total length of all edges that get mapped to \(T_i\) fit.
\end{proof}

\begin{proof}[\cref{thm:mcis-mln}]
	We proceed analogously to the proof of \cref{thm:mcs-mln} and only describe the differences here. Again, we begin by guessing a graph \(S\) with the same properties. Note that this time we cannot assume \(S\) to have no isolated vertices, as we are in the induced setting. Thus, when we guess the mapping \(\eta_i\) from \(S\) to~\(G_i\), we additionally guess for each vertex \(v \in V(S)\) with \(\deg(v)=0\), either a vertex \(\eta_i(v)\in V(G_i)\) with \(\deg(\eta_i(v)) \neq 2\), or guess a maximal {degree\mbox{-}2} trail~\(\eta_i(v) \in \mathcal{T}_2(G_i).\) It remains the same that we afterwards ensure that all vertices in \(V(G_i)\) are guessed at most once, and further all inner elements of~\(\eta_i(E(S))\) are disjoint to all elements in \(\eta_i(V(S))\).
	
	When guessing \(S\) we additionally demand the following:
	If \(\eta_i(u),\eta_i(v) \in V(G_i)\) are adjacent, then there is the edge \(\{u,v\} \in E(S)\). Otherwise we reject this guess.
	This ensures, that all vertices of the resulting graph with degree greater than~2 can be embedded in an induced way. Thus, it remains to ensure that all vertices with degree not greater than~2 are embedded correctly. We do this by adapting the ILP. Note that we can consider the same objective function as in the non-induced case, because the number of non-degree-2 vertices (including isolated vertices) is already fixed by the choice of \(S\) and hence, the number of vertices in the resulting graph is maximized if and only if the sum over the \(l_e\) is maximized.
	
	First of all, we also need to introduce variables \(v_{x_i}, v_{y_i}\) and \(d_{u,z_i}\) for the isolated vertices of \(S\) which get mapped by \(\eta_i\) to a maximal degree-2 trail. They also need to satisfy constraints (\ref{cons:1}) to (\ref{cons:3}). Then, we further include the following two additional constraints to the ILP:
	\begin{enumerate}
		\item[N1.]\label{const:N1} For all \(v_{x_i}\) such that there exists some \(u \in V(S)\) with \(\eta_i(u)=x_i\) and \(\{u,v\} \notin  E(S)\):
		\[v_{x_i} \geq 2.\]
		\item[N2.]\label{const:N2} For all \(d_{v,z_i}\) such that there exists some \(u \in V(S)\) with \(\eta_i(u)=z_i\) and \(\{u,v\} \notin  E(S)\):
		\[d_{v,z_i} \geq 2.\]
	\end{enumerate}
	These two additional constraints ensure that only positions for the degree-0 and degree-1 vertices of the resulting graph are considered, which can correspond to an induced embedding.
	
	Finally, we replace the equation in constraint (\ref{cons:length}) in the ILP with
	\[|E(T_i)| \geq (a+2) + b + 2k + \sum_{f \in E(S) \text{ with } \eta_i(f)= (T_i)} (l_{f}+2),\]
	where \(k\) is the number of isolated vertices \(v \in V(S)\) with \(\eta_i(v) = (T_i).\)
	
	This change in constraint (\ref{cons:length}) ensures, that within each trail \(T_i\) of \(G_i\), there is enough space to embed the different components which get mapped to \(T_i\) in a way such that they are not adjacent. Since all components which get mapped to~\(T_i\) are isolated vertices or paths (in the non-induced case they were all paths), it suffices again to check if the length of the trails fit.
	
	Hence, it remains to ensure that all isolated loops can be embedded in an induced way. The isolated loops \(e=\{u,u\}\) for which we guessed \(\eta_i(u) \in V(G_i)\) are covered by the additional assumption on the initial guessing of \(S\) and the new ILP\@. All other isolated loops get embedded into an isolated cycle in~\(G_i\) of correct length (ensured by constraint \ref{const:loops}) and hence, are embedded as an induced subgraph.
\end{proof}


\section{{\probMCSshort} and {\probMCISshort} parameterized by neighborhood diversity}

The \emph{neighborhood diversity} of a graph $G$, denoted $\nd(G)$, is the number of twin classes in $G$.
Clearly, the twin classes of a graph (and thus its neighborhood diversity as well) can be computed in polynomial time.
Recall that a twin class is a clique or an independent set.

By the definition of twins, the connection between two twin classes is either \emph{full} or \emph{empty};
that is, there are either no or all possible edges between them.

\begin{theorem}
\label{thm:mcis-nd}
{\probMCIS} is fixed-parameter tractable parameterized by $\nd(G_{1}) + \nd(G_{2})$. 
\end{theorem}
\begin{proof}
We solve {\probMCISshort} parameterized by neighborhood diversity 
by solving instances of \textsc{Integer Linear Programming} (ILP) parameterized by the number of variables~\cite{Lenstra83} (see also \cite{FellowsLMRS08}).

Let $\langle G_{1}, G_{2}, h \rangle$ be an instance of {\probMCIS}.
Let $U^{(1)}_{1}, \dots, U^{(1)}_{p}$ and $U^{(2)}_{1}, \dots, U^{(2)}_{q}$ be the twin classes of $G_{1}$ and~$G_{2}$, respectively.
For $i \in \{1,2\}$ and $h \le h'$, we say that $(U^{(i)}_{h}, U^{(i)}_{h'})$ is an \emph{adjacent pair} if 
either $h = h'$ and $U^{(i)}_{h}$ is a clique
or $h < h'$ and the connection between $U^{(i)}_{h}$ and~$U^{(i)}_{h'}$ is full.

We are going to find a maximum common induced subgraph $H$ of~$G_{1}$ and~$G_{2}$ 
with an induced subgraph isomorphism $\eta_{i}$ from $H$ to $G_{i}$ for each~$i \in \{1,2\}$.

For $i \in [p]$ and $j \in [q]$, let $X_{i,j} = V(H) \cap \eta_{1}^{-1}(U^{(1)}_{i}) \cap \eta_{2}^{-1}(U^{(2)}_{j})$.
Note that~$X_{i,j}$ is a (not necessarily maximal) set of twins of $H$.
For $i \in [p]$ and $j \in [q]$, we take a non-negative integer variable $x_{i,j}$ that represents $|X_{i,j}|$.
We add the size constraints $\sum_{j \in [q]} x_{i,j} \le |U^{(1)}_{i}|$ for all $i \in [p]$
and $\sum_{i \in [p]} x_{i,j} \le |U^{(2)}_{j}|$ for all~$j \in [q]$.
We set the sum of $x_{i,j}$ to the objective function to be maximized; i.e., the objective is:
\[
 \textrm{maximize} \quad \sum_{i \in [p], \; j \in [q]} x_{i,j}.
\]

Now we branch to $3^{p q}$ instances of ILP by adding one of the constraints $x_{i,j} = 0$, $x_{i,j} = 1$, or $x_{i,j} \ge 2$ for $i \in [p]$ and $j \in [q]$.
We reject this guess  if at least one of the following holds:
\begin{itemize}
  \item $x_{i,j} \ge 2$ for some $i \in [p]$ and $j \in [q]$, 
  where one of $U^{(1)}_{i}$ and $U^{(2)}_{j}$ is an independent set and the other is a clique;
  \item $x_{i,j} \ne 0$ and $x_{i',j'} \ne 0$ for some $(i,j), (i',j') \in [p] \times [q]$ with $(i,j) \ne (i',j')$,
  where exactly one of $(U^{(1)}_{i}, U^{(1)}_{i'})$ and $(U^{(2)}_{j}, U^{(2)}_{j'})$ is an adjacent pair.
\end{itemize}
In the former case, $G_{1}$ and $G_{2}$ disagree in $X_{i,j}$.
In the latter, $G_{1}$ and $G_{2}$ disagree at the connection between $X_{i,j}$ and $X_{i',j'}$ as one is full but the other is empty.
On the other hand, we can see that
if the guess is not rejected, then the assignment to the variables $x_{i,j}$ correctly represents a common induced subgraph of \(G_1\) and~\(G_2\)
with $\sum_{i \in [p], \; j \in [q]} x_{i,j}$ vertices.

We solve all $3^{p q}$ instances of ILP with $p q$ variables and output the largest solution found. 
Since $p = \nd(G_{1})$ and $q = \nd(G_{2})$, the theorem follows.
\end{proof}

\begin{theorem}
\label{thm:mcs-nd}
{\probMCS} is fixed-parameter tractable parameterized by $\nd(G_{1}) + \nd(G_{2})$.
\end{theorem}
\begin{proof}
We reduce {\probMCSshort} parameterized by neighborhood diversity 
to \textsc{Integer Quadratic Programming} (IQP) parameterized by the number of variables.
Given $Q \in \mathbb{Z}^{n \times n}$, $c \in \mathbb{Z}^{n}$, $A \in \mathbb{Z}^{m \times n}$, $b \in \mathbb{Z}^{m}$,
IQP asks to find a vector $x \in \mathbb{Z}^{n}$ that maximizes $x^{T} Q x + c^{T} x$ subject to $A x \le b$.
It is known that IQP parameterized by the number of variables (i.e., $n$) plus the maximum absolute value of a coefficient in~$A$, $Q$, and $c$
is fixed-parameter tractable~\cite{Lokshtanov15,Zemmer17}.\footnote{
Although the objective function of IQP is usually represented by a quadratic form $x^{T} Q x$ only,
having a linear term $c^{T} x$ can be done without changing the fixed-parameter tractability. See the discussion by Lokshtanov~\cite[p.~4]{Lokshtanov15}.}

Let $\langle G_{1}, G_{2}, h \rangle$ be an instance of {\probMCS}.
Let $U^{(1)}_{1}$, $\ldots$, $U^{(1)}_{p}$ and $U^{(2)}_{1}, \dots, U^{(2)}_{q}$ be the twin classes of $G_{1}$ and $G_{2}$, respectively.
We define adjacent pairs of twin classes as in the same way in the proof of \cref{thm:mcis-nd}.
We are going to find a maximum common subgraph $H$ of $G_{1}$ and~$G_{2}$ 
with a subgraph isomorphism $\eta_{i}$ from $H$ to $G_{i}$ for each $i \in \{1,2\}$.
For $i \in [p]$ and $j \in [q]$, let $X_{i,j} = V(H) \cap \eta_{1}^{-1}(U^{(1)}_{i}) \cap \eta_{2}^{-1}(U^{(2)}_{j})$.
For $i \in [p]$ and $j \in [q]$, we take a non-negative integer variable $x_{i,j}$ that represents $|X_{i,j}|$.
We add the size constraints $\sum_{j \in [q]} x_{i,j} \le |U^{(1)}_{i}|$ for all $i \in [p]$
and $\sum_{i \in [p]} x_{i,j} \le |U^{(2)}_{j}|$ for all~$j \in [q]$.

Observe that the 
number of edges in $H[X_{i,j}]$ is $\binom{x_{i,j}}{2}$ if both $U^{(1)}_{i}$ and $U^{(2)}_{j}$ are cliques, and $0$ otherwise. 
Also, for $(i,j) \ne (i', j')$, the number of edges between~$X_{i,j}$ and $X_{i',j'}$ in $H$ is $x_{i,j} \cdot x_{i',j'}$
if both $(U^{(1)}_{i}, U^{(1)}_{i'})$ and $(U^{(2)}_{j}, U^{(2)}_{j'})$ are adjacent pairs, and $0$ otherwise.
In total, the number of edges in~$H$ is
\[
|E(H)| = \sum_{\textrm{cliques } U^{(1)}_{i},\; U^{(2)}_{j}} \binom{x_{i,j}}{2}
+ \sum_{\substack{\textrm{adjacent pairs } (U^{(1)}_{i}, U^{(1)}_{i'}), \; (U^{(2)}_{j}, U^{(2)}_{j'}), \\ (i,j) \ne (i',j')}} x_{i,j} \cdot x_{i',j'}.
\]
Recall that when a twin class $W$ is a clique, then $(W,W)$ is an adjacent pair.
To make the coefficients integral, we set the objective function to $2|E(H)|$. 

{\spaceskip=3pt plus 1pt The number of variables in the obtained IQP instance is $pq$
and the maximum absolute value in the objective function and the left-hand side of the constraints is constant. This implies the theorem.}
\end{proof}


\section{Conclusion}

In this paper, we showed fixed-parameter tractable cases for {\probMCSshort} and {\probMCISshort}.
Given our results, the parameterized complexity of these problems with respect to well-studied structural parameters is almost completely understood (see \cref{fig:graph-parameters}).
Filling the missing part (i.e., {\probMCISshort} parameterized by $\cvd(G_{1})+\cvd(G_{2})$) would be the natural next step.

\appendix

\section{Tight bound on the number of non-degree-2 vertices}

In this section, we prove an upper bound on the number of non-degree-2 vertices in terms of the max-leaf number.
Our bound is tight for paths with two or more vertices (and more generally, for path forests with no isolated vertices).

\begin{lemma}
\label{lem:deg-2-bound-mln}
For every graph $G$ with at least two vertices, $|V_{\ne 2}(G)| \le 4 \ml(G)-6$.
\end{lemma}
\begin{proof}
	As a first step, we prove the statement for connected graphs without degree-2 vertices. Namely, we prove the following claim.
	\begin{myclaim}
		\label{clm:mln-deg2-bound-solvingdeg2}
		If a connected graph \(G\) without degree-2 vertices has at least two vertices, then \(|V(G)| \le 4 \ml(G)-6\).
	\end{myclaim}
	\begin{subproof}
		[\cref{clm:mln-deg2-bound-solvingdeg2}]
		We assume that $G$ has maximum degree at least~$3$ as otherwise $G$ is $K_{2}$.
		We also assume that $G$ has degree-$1$ vertices since otherwise we can directly apply 
		the result of Kleitman and West~\cite[Theorem~2]{KleitmanW91}, 
		who showed that $|V(G)| \le 4\ml(G) - 8$ if $G$ has minimum degree at least~$3$.
		We closely follow their proof and make some minor modifications to handle degree-$1$ vertices.
		
		We begin with a small subtree $T$ of $G$ and expand $T$ iteratively to obtain a spanning tree of $G$.
		We denote by $n$ and $\ell$ the (current) numbers of vertices and leaves of $T$, respectively.
		For a leaf $x$ of $T$, let~$d'(x)$ denote the \emph{out-degree} $|N_{G}(x) \setminus V(T)|$ of $x$. Namely
		$d'(x)$ is the number of neighbors of $x$ not in $T$.
		The \emph{expansion} at $x$ enlarges $T$ by adding the $d'(x)$ edges from $x$ to all its neighbors not in $T$.
		We grow $T$ by a sequence of expansions. This guarantees that $d'(v) = 0$ for each inner vertex $v$ of $T$.
		In other words, only leaves of $T$ may be adjacent to vertices not in $T$. We say that a leaf $x$ of $T$ is \emph{dead} if $d'(x) = 0$. Let $m$ be the number of dead leaves in $T$.
		An \emph{admissible operation} is a single expansion or a sequence of two expansions that satisfies
		\begin{equation}
			3 \D \ell + \D m \ge \D n,
			\label{eq:augmentation-ineq}
		\end{equation}
		where $\D \ell$, $\D m$, and $\D n$ are the increases of the numbers of leaves, dead leaves, and vertices, respectively, by the operation.
		We call \cref{eq:augmentation-ineq} the \emph{augmentation inequality}.
		
		\paragraph{Initializing $T$.}
		
		Let $u$ be a degree-$1$ vertex of $G$ and $v$ be its neighbor.
		Observe that $\deg_{G}(v) \ge 3$ as~$G \ne K_{2}$.
		Let $\mu$ be the number of degree-$1$ vertices adjacent to $v$.
		We initialize $T$ as the star centered at $v$ with all its neighbors as the leaves.
		Now $\ell = \deg_{G}(v)$, $m = \mu \ge 1$, and $n = \deg_{G}(v) + 1$.
		
		\paragraph{Final form of $T$.}
		Assume that we successfully applied admissible operations to make $T$ a spanning tree of $G$ with $L$ leaves.
		Now all leaves of $T$ are dead and~$n = |V(G)|$.
		By summing the augmentation inequality for all operations that we applied, we get
		\[
		3 (L-\deg_{G}(v)) + (L-\mu) \ge |V(G)| - (\deg_{G}(v) + 1),
		\]
		which simplifies to $|V(G)| \le 4L -2\deg_{G}(v) -\mu +1$.
		Since $\deg_{G}(v) \ge 3$ and~$\mu \ge 1$,  we have $|V(G)| \le 4L-6$.

		\paragraph{Analyzing each step.}
		It remains to show that we can always apply an admissible operation to $T$ unless it is already a spanning tree (i.e., as long as there is a leaf with positive out-degree).
		
		\textit{Case~1:} If there is a leaf $x$ with $d'(x) \ge 2$, then the expansion at $x$ yields $\D \ell = \D n -1 = d'(x) - 1\ge 1$ and $\D m \ge 0$.
		
		In the following cases, we assume that all leaves of $T$ have out-degree at most~$1$.
		
		\textit{Case~2:} Assume there is a leaf $x$ with $d'(x) = 1$ such that its neighbor $y$ not in $T$ has degree at least~$3$ in~$G$.
		\begin{itemize}
			\item Case~2(a): If $y$ has at least two neighbors in $T$,
			then the expansion at $x$ yields $\D \ell = 0$, $\D m \ge 1$, and $\D n = 1$.
			To see that $\D m \ge 1$, observe that the neighbors of $y$ in $T$, except for $x$, become dead leaves after this expansion.
			(Recall that only leaves of $T$ may be adjacent to vertices not in $T$.)
			
			\item Case~2(b): If $y$ has at least two neighbors not in $T$,
			then the sequence of expansions at $x$ and then at $y$ yields $\D \ell = \D n - 2 \ge 1$ and $\D m \ge 0$.
		\end{itemize}
		
		\textit{Case~3:} If there is a leaf $x$ with $d'(x) = 1$ such that its neighbor $y$ not in $T$ has degree~$1$ in~$G$, 
		then the expansion at $x$ yields $\D \ell = 0$, $\D m = 1$, and $\D n = 1$.
		
		This shows that there always exists a sequence of admissible operations to make $T$ a spanning tree of $G$.
	\end{subproof}
	 
	 For connected graphs, to deal with degree-2 vertices, we do an induction on the number of vertices and edges in the graph. Let \(G\) be a connected graph on~\(n\) vertices and \(m\) edges. Assume \(G\) has a degree-2 vertex \(v\), as otherwise the statement follows from \cref{clm:mln-deg2-bound-solvingdeg2}.
	 
	 \textit{Case 1:} Assume \(v\) is a vertex-separator of \(G\). \\
	 In this case, \(v\) is part of a bridge (i.e., an edge whose removal separates the graph) \(\{u,v\} \in E(G)\).
	 Contract the edge \(\{u,v\}\) to obtain a new connected graph \(H\) with \(n-1\) vertices. Then \(\ml(G) = \ml(H)\), since every spanning tree in \(G\) contains \(\{u,v\}\) and thus, every spanning tree in \(H\) corresponds to one in \(G\) with the same number of leaves. Furthermore, \(|V_{\ne 2}(G)| = |V_{\ne 2}(H)|\) since \(\deg_G(u) = \deg_H(w_{uv})\) and \(\deg_G(v) = 2\), where $w_{uv}$ is the vertex that replaced the edge $\{u,v\}$.  Thus, by applying the induction hypothesis we obtain 
	 \[
	 |V_{\ne 2}(G)| = |V_{\ne 2}(H)| \leq 4\ml(H) -6  = 4\ml(G) -6.
	 \]
	 
	 \textit{Case 2:} Assume \(v\) is not a vertex-separator of \(G\).  \\
	 Let \(u\) be one of the two neighbors of \(v\). Delete the edge \(\{u,v\}\) from \(G\) to obtain a new graph \(H\) with \(m-1\) edges. Then \(|V_{\ne 2}(H)| \in \{|V_{\ne 2}(G)|, |V_{\ne 2}(G)|+1\}\), since~\(v\) has degree 2 in \(G\) and degree 1 in \(H\), and \(u\) has degree \(d>1\) in \(G\) and degree~\(d-1\) in~\(H\). In particular, \(|V_{\ne 2}(H)| \geq |V_{\ne 2}(G)|\). Since \(H\) is a connected subgraph of \(G\), we have \(\ml(H) \leq \ml(G)\). Thus, by applying the induction hypothesis we obtain
	 \[
	 |V_{\ne 2}(G)| \leq |V_{\ne 2}(H)| \leq 4\ml(H)-6 \leq 4\ml(G)-6.
	 \]
	 This finishes the proof of the statement for connected graphs. 

Assume that $G$ is disconnected.
Let $\cc_{2}(G)$ be the set of connected components of $G$ with at least two vertices.
If $\cc_{2}(G) = \emptyset$, then $|V_{\ne 2}(G)| = |V(G)| = \ml(G) \le 4 \ml(G)-6$,
where the last inequality holds since $\ml(G) = |V(G)| \ge 2$.
If $\cc_{2}(G) \ne \emptyset$, then we can apply the lemma for the connected case to each nontrivial component in $\cc_{2}(G)$ as follows:
\begin{align*}
  |V_{\ne 2}(G)| 
  &= 
  \sum_{C \in \cc(G)} |V_{\ne 2}(C)| \le \sum_{C \in \cc_{2}(G)} (4\ml(C)-6) + |\cc(G) \setminus \cc_{2}(G)|
  \\ 
  &\le 
  4 \Bigg(\sum_{C \in \cc(G)} \ml(C) \Bigg) -6 = 4\ml(G)-6.
\end{align*}
This completes the proof.
\end{proof}


%
%
\bibliographystyle{plainurl}
\bibliography{ref}


\end{document}